\documentclass[11pt]{article}

\usepackage{amsmath, amssymb, amsthm, bm, latexsym, color, graphicx, amscd, hyperref, multicol, textgreek, comment, tikz}
\usetikzlibrary{shapes.geometric, arrows}
\usepackage{soul}

\numberwithin{equation}{section}

\usepackage[margin=2.8cm]{geometry}

  \newcommand{\F}{\mathbb{F}}

\newcommand{\myalpha}{\alpha} \newcommand{\mybeta}{\beta} \newcommand{\mytau}{\tau}
   
\newcommand{\gauss}[3]{{#1 \brack #2}_#3}

\theoremstyle{definition}

\newtheorem{theorem}{Theorem}[section] \newtheorem{lemma}[theorem]{Lemma}
\newtheorem{corollary}[theorem]{Corollary} \newtheorem{proposition}[theorem]{Proposition}

 \newtheorem{remark}[theorem]{Remark}

\newcommand{\prp}[1]{\textcolor{purple}{#1}}

\tikzstyle{startstop} = [rectangle, rounded corners, minimum width=3cm, minimum height=1cm, text centered, text width=3cm, draw=black, fill=red!30]
\tikzstyle{io} = [trapezium, trapezium stretches=true, trapezium left angle=70, trapezium right angle=110, minimum width=3cm, minimum height=1cm, text centered, draw=black, fill=blue!30]
\tikzstyle{process} = [rectangle, minimum width=3cm, minimum height=1cm, text centered, text width=3cm, draw=black, fill=orange!30]
\tikzstyle{decision} = [diamond, minimum width=3cm, minimum height=1cm, text centered, draw=black, fill=green!30]
\tikzstyle{arrow} = [thick,->,>=stealth]

\newcommand{\rmv}[1]{}

\newcommand{\myalphatilde}{\tilde{\alpha}}
\newcommand{\mybetatilde}{\tilde{\beta}}
\begin{document}

\title{{The Random Variables of the DNA Coverage Depth Problem}}

\author{\c{S}eyma Bodur \and Stefano Lia \and 
Hiram H. L\'opez \and Rati Ludhani \and Alberto Ravagnani \and Lisa Seccia}
\date{}

% \keywords{DNA storage; coverage depth problem; random access; MDS; Hamming; simplex; expected value; variance; higher moments}

% \subjclass[2010]{94B05;}

\maketitle

\begin{abstract}
DNA data storage systems encode digital data into DNA strands, enabling dense and durable storage. Efficient data retrieval depends on coverage depth, a key performance metric. We study the random access coverage depth problem and focus on minimizing the expected number of reads needed to recover information strands encoded via a linear code. 
We compute the asymptotic performance of a recently proposed code construction, establishing and refining a conjecture in the field by giving two independent proofs.
We also analyze a geometric code construction based on balanced quasi-arcs and optimize its parameters. Finally, we investigate the full distribution of the random variables that arise in the coverage depth problem, of which the traditionally studied expectation is just the first moment. This allows us to distinguish between code constructions that, at first glance, may appear to behave identically.
\end{abstract}

\bigskip

\section{Introduction}
DNA data storage is a promising solution to the exponentially increasing demand for sustainable and scalable data storage~\cite{BarLev2025CoverYB}. 
DNA offers significant advantages over traditional storage media~\cite{Lee2019, Shomorony2022},
due to its density, durability, and sustainability. 
A typical DNA storage system consists of three main processes: \emph{synthesis}, \emph{storage}, and \emph{sequencing}~\cite{BarLev2025CoverYB}. First, digital data is encoded into sequences over the DNA alphabet and synthesized into DNA strands. These strands are then stored in an unordered fashion within a container. When data needs to be retrieved, the 
stored strands are randomly sequenced
and converted back to digital sequences~\cite{VANDIJK2018666, Weirather2017}. 
When attempting to retrieve data by random sampling,
multiple copies of each strand are obtained, many of which will turn out to be useless.
This issue is captured by the notion of \emph{coverage depth}, defined as the ratio between the number of sequenced reads and the number of synthesized DNA strands~\cite{Heckel2019}. In recent work \cite{BarLev2025CoverYB}, the \emph{random access coverage depth problem} was introduced to formally characterize the efficiency of data recovery from DNA.

The random access coverage depth problem, which is the focus of this paper, can be briefly illustrated as follows.
Assume that $k$ information strands (representing the data) are encoded into $n$ strands via a matrix $G \in \F_q^{k \times n}$, often interpreted as the generator matrix of an error-correcting code. The columns of $G$ are drawn uniformly at random and with repetition.  
What is the expected number of draws needed for the standard basis vector~$e_i$ to be in the span of the drawn columns? What is the maximum of these expectations, say $\mathbb{E}[G]$, as $i \in \{1, \ldots, k\}$? And finally, what is the smallest value that can be obtained for such a maximum as $G$ ranges over the rank $k$ matrices in $\F_q^{k \times n}$?

The questions just stated are typically difficult to answer and have attracted growing interest in bounds and  constructions for matrices~$G$~\cite{BarLev2025CoverYB,GruicaMontanucciZullo2024,GruicaBarLevRavagnaniYaakobi2024,coveringallbases,boruchovsky2025makingfirstrandomaccess,bertuzzoravagnani,cao2025optimizingsequencingcoveragedepth}. Even when such constructions perform well in practice, precisely evaluating their performance and computing the value of $\mathbb{E}[G]$ defined above can remain challenging. That is the case of two interesting constructions proposed in~\cite{BarLev2025CoverYB} and~\cite{GruicaMontanucciZullo2024}, which we analyze in this paper. 

The code construction of~\cite{BarLev2025CoverYB} is given by a $k \times 2k$ matrix. While the matrix has a relatively simple structure, computing $\mathbb{E}[G]$ is a nontrivial task. The authors derive a recursive description and based on experimental results conjecture that $\lim_{k \to +\infty} \mathbb{E}[G]/k<0.9456$, a conjecture that was later established in~\cite{GruicaMontanucciZullo2024}, without however computing the exact value of the limit and thus the code's performance.
The code construction of~\cite{GruicaMontanucciZullo2024} uses instead balanced quasi-arcs, mathematical objects inspired by finite geometry. The construction depends on the choice of two parameters,~$x$ and~$y$. In this context, the challenging task is to compute the ratio $y/x$ offering the best performance.

We now briefly describe how the rest of the paper is organized, which also gives us the chance to highlight the main contribution made by this work.
In Section~\ref{a particular code}, we focus on the conjecture proposed in~\cite{BarLev2025CoverYB}. We give a closed expression for the limit $\lim_{k \to +\infty} \mathbb{E}[G]/k$, offering two independent proofs. The first proof completely solves the recurrence relation from~\cite{BarLev2025CoverYB} using an argument inspired by the matrix exponentiation method; the second proof uses a stand-alone argument based on combinatorial identities.
The limit value we obtain is a curious expression involving $\pi$ and $\sqrt{3}$.

In Section~\ref{quasi-arcs}, we concentrate on the construction from~\cite{GruicaMontanucciZullo2024} based on
balanced quasi-arcs. We are able to give a closed formula for the limit of the expectation as a function of the ratio $y/x$. As an application, we 
compute the optimal value of $y/x$. 

Finally, in Section~\ref{higher}, we consider the problem of computing the higher moments of the random variables that arise in the coverage depth random access problem.
We show that the quantities often used to compute the expectation $\mathbb{E}[G]$ actually determine all moments.
In Section~\ref{famous codes}, we apply the results from Section~\ref{higher} to show how to assess the performance of codes where the value of $\mathbb{E}[G]$ is the same, most notably 
MDS, Hamming, and Simplex codes.

\paragraph*{Acknowledgements.}
This work was initiated during the ``Coding Theory and Cryptography Summer School and Collaboration Workshop'', held at the Stager Center for International Scholarship -- Virginia Tech, Switzerland, in July 2024. The authors are very grateful to the organizers of the event for facilitating this collaboration.
\c{S}eyma Bodur was partially supported by grants PID2022-138906NB-C21 funded by MICIU/AEI/10.13039/501100011033 and by ERDF/EU, and by grant CONTPR-2019-385 funded by Universidad de Valladolid and Banco Santander.
Stefano Lia was partially supported by the Irish Research Council, grant n. GOIPD/2022/307.
Hiram H. L\'opez was partially supported by the NSF grant DMS-2401558.
Rati Ludhani was supported by Prime Minister’s Research Fellowship PMRF-192002-256 at IIT Bombay and acknowledges the NBHM travel grant to attend the event above.
Alberto Ravagnani was partially supported by the Dutch Research Council via grant OCENW.KLEIN.53.
Lisa Seccia was partially supported by SNSF grant TMPFP2\_217223.

\section{Expectation of the rate 1/2 code from\texorpdfstring{~\cite{BarLev2025CoverYB}}{ [BarLev2025CoverYB] }}
\label{a particular code}

In this section, we concentrate on a code construction proposed in~\cite{BarLev2025CoverYB}, computing the asymptotic performance of the code and establishing a conjecture of the authors.
We start by establishing the notation for the rest of the paper.

In the sequel, $G \in \F_q^{k \times n}$ is a rank $k$ matrix over the finite field $\F_q$ with $q$ elements. 
For $1\le j\le n$, fix the notation $G^j$ to denote the $j$-th column of $G$.

We draw the columns of $G$ uniformly at random and with repetition. For $i \in \{1, \ldots, k\}$, let $\tau_i(G)$ be the random variable that governs the number of columns of $G$ that are drawn until the standard basis vector $e_i$ is in their $\F_q$-span.
The variable $\tau_i(G)$ therefore models the time of retrieval of the $i$-th information strand. 
The random access coverage depth problem is to compute or estimate the expectation \( \mathbb{E}[\tau_i(G)] \).

We now turn to the code of \cite[Construction~1]{BarLev2025CoverYB}, specified by a $k\times 2k$ generator matrix $G_{k\times2k}$ whose columns are $e_1,e_2,\dots,e_k,e_1+e_2,e_2+e_3,\dots,e_{k-1}+e_{k},e_k+e_1$, where $e_1,\ldots,e_k$ are the standard basis vectors of $\F_q^k$:
$$G_{k \times 2k} =
\begin{pmatrix}
1 & 0 & \cdots & 0 & 1 & 0 & \cdots & 0 & 1 \\
0 & 1 & \cdots & 0 & 1 & 1 & \cdots & 0 & 0 \\
\vdots & \vdots & \ddots & \vdots & \vdots & \vdots & \ddots & \vdots & \vdots \\
0 & 0 & \cdots & 1 & 0 & 0 & \cdots & 1 & 1 \\
\end{pmatrix}
$$

In~\cite{BarLev2025CoverYB}, the authors observe with computational results that
$\smash{\mathbb{E}(\tau_i(G_{k\times2k}))<k}$ for all $i$. It is not difficult to see that
$\mathbb{E}(\tau_i(G_{k\times2k}))$ does not depend on $i$. Furthermore, the authors of~\cite{BarLev2025CoverYB}
conjecture that the value $\smash{\ell_k=\mathbb{E}(\tau_i(G_{k\times2k}))/k}$ decreases when $k$ increases, with $\ell=\lim_{k\rightarrow\infty}\ell_k<0.9456$; see~\cite[Conjecture 1]{BarLev2025CoverYB}.
The conjecture was later investigated in \cite{GruicaMontanucciZullo2024}, where the authors prove the following result.

\begin{theorem}[\text{\cite[Corollary 5.9]{GruicaMontanucciZullo2024}}]
    We have $\limsup_{k \to \infty} \ell_k \le \frac{70318847}{74364290} \approx 0.945599655$.
\end{theorem}
Note that the previous result suffices to answer the open question in~\cite{BarLev2025CoverYB}, although it does not prove the limit exists. Furthermore, the proof
in~\cite{GruicaMontanucciZullo2024} relies on lengthy calculations.

In this paper, we take a completely different approach, which allows us to prove that
$\lim_{k \to \infty} \ell_k$ exists and to give a simple, closed expression for its value. The result is a curious expression involving both radicals and $\pi$. More precisely, we will prove the following.

\begin{theorem}\label{thm:conjproof}
For all $i\in \{1, \ldots, k\}$, we have
\[
\lim_{k \to \infty} \frac{\mathbb{E}[\tau_i (G_{k\times2k})]}{k}=\frac{8\sqrt{3}\pi-18}{27}.
\]
\end{theorem}

We provide two proofs for Theorem~\ref{thm:conjproof}. The first builds on a result from~\cite{BarLev2025CoverYB}, while the second is a standalone argument. We include both arguments as they offer different approaches and tool sets.
The result from~\cite{BarLev2025CoverYB} that we need for the first proof is the following.

\begin{theorem}\cite[Theorem 10]{BarLev2025CoverYB}\label{thm:expC(2k,k)}
For any $k \geq 2$ and any $i \in \{1,\dots,k\}$, we have
\begin{align}\label{eq:BSGYformula}
\mathbb{E}[\tau_i (G_{k\times2k})] &= 1 + \sum_{j=1}^{2k-3} B(k-1, j) \, \frac{2k}{(2k - j)\binom{2k}{j}},
\end{align}
where
\begin{align*}
B(k, j) = \begin{cases}
\binom{2k-1}{j} + 2B(k-1, j-1) - B(k-2, j-2) & \text{if } k \geq 2, j \geq 2,\\
1 & \text{if } k \geq 0, j = 0\text{ or } k = 1, j = 2,\\
2k + 1 & \text{if } k \geq 0, \, j = 1,\\
0 & \text{if } k = 0, \,  j \geq 2 \text{ or }k = 1, j \geq 3. \\
\end{cases}
\end{align*}
\end{theorem}

Note that there is a typo in the formula for the expected value $\mathbb{E}[\tau_i (G_{k\times2k})]$ in \cite[Theorem~10]{BarLev2025CoverYB}. The correct term in the summation is $B(k-1, j)$, not $B(k, j)$.

The approach of~\cite{GruicaMontanucciZullo2024}
also relies on the previous theorem,
giving a refined first-order recurrence for
$B(k,j)$ and using some clever approximations. 
In the following lemma, we show how to obtain the exact value of $B(k,j)$ with a straightforward application of the matrix exponentiation method.

\begin{lemma}\label{lem:B(k,i)}
Let $B(k,j)$ be as in Theorem~\ref{thm:expC(2k,k)}. We have
\[
B(k,j)=\sum_{t=1}^{min\{k-1,\, j-1\}}\binom{2(k-t+1)-1}{j-t+1}\, t+\sigma(j,k),
\]
where $\sigma(j,k)=(2(k-j+1)+1)j -(j-1)$ for $j\leq k$, $\sigma(k+1,k)=1$, and $\sigma(j,k)=0$ for $j\geq k+2$.
\end{lemma}
\begin{proof}
We define two vectors in $\mathbb{R}^2$ as follows: $$\tilde{B}{(k,j)} = {\begin{pmatrix}
   B{ (k,j)} \\
  B{ (k-1,j-1)}
\end{pmatrix}} \quad \text{ and } \quad 
F{(k,j)} = \begin{pmatrix}
   \binom{2k-1}{j} \\
  0 
\end{pmatrix}.$$ 
With this notation, for $k\geq 2$ and $j\geq 2$ we write the recursion as
\[
\tilde{B}{ (k,j)}=
A
\tilde{B}{ (k-1,j-1)}
+
F{(k,j)}=
A(A
\tilde{B}{ (k-2,j-2)}
+
F{ (k-1,j-1)})
 +F{ (k,j)}=..., 
\]
where 
\[
A=
\begin{bmatrix}
  2 & -1 \\
  1 & 0 
\end{bmatrix}\quad\text{ and }\quad
A^l=
\begin{bmatrix}
  l+1 & -l \\
  l & 1-l 
\end{bmatrix}.
\]
Expanding the expression above, we obtain
\begin{equation*}
\tilde{B}(k,j)=
\begin{cases}
F{ (k,j)}+
AF{ (k-1,j-1)}
+\cdots+
%A^{k-2}F{ (k-(k-2),j-(k-2))}+
A^{k-1}\tilde{B}{(k-(k-1),j-(k-1))}
&\text{if }\quad j\geq k+1,\\
F{  (k,j)}+
AF{ (k-1,j-1)}
+\cdots+
%A^{i-2}F{ (k-(i-2),j-(i-2))}+
A^{j-1}\tilde{B}{(k-(j-1),j-(j-1))}& \text{if }\quad j\leq k.
\end{cases}
\end{equation*}
The base cases $\tilde{B}{ (k,1)}$ and $\tilde{B}{ (1,j)}$ are as in Theorem \ref{thm:expC(2k,k)}.
The result is now the equality obtained by considering only the first entry of each vector. 
\end{proof}

We can now establish Theorem~\ref{thm:conjproof}, providing a complete answer to the open question from~\cite{BarLev2025CoverYB}.

\begin{proof}[Proof of Theorem~\ref{thm:conjproof}]
We substitute the value of $B(k-1,i)$ computed by Lemma~\ref{lem:B(k,i)} in~\eqref{eq:BSGYformula} and
exchange the order of summation. We obtain 

\begin{eqnarray*}
&\displaystyle \frac{\mathbb{E}[\tau_j(G_{k\times2k})]}{k}
& =\frac{1}{k}\left(\sum_{t=1}^{k-2} \sum_{i=t+1}^{2k-3} t\binom{2k-2t-1}{i-t+1} \frac{1}{\binom{2k-1}{i}}+\sum_{i=1}^{k-1}  \frac{2(k-i)i+1}{\binom{2k-1}{i}}\right)\\
&& =\frac{1}{k}\left(\sum_{t=1}^{k-2} \sum_{i=t+1}^{2k-3} \frac{t \binom{2k-2t-1}{i-t+1}}{\binom{2k-1}{i}}+\sum_{i=1}^{k-1}  \frac{2(k-i)i}{\binom{2k-1}{i}}\right).
\end{eqnarray*}
We investigate the two summations separately. On the one hand, we have
\[
\lim_{k\to \infty} \frac{1}{k}\sum_{i=1}^{k-1}  \frac{2(k-i)i}{\binom{2k-1}{i}}=0,
\] 
since the summand is bounded from above by $\smash{\frac{2ki^{i+1}}{(2k-1)^i}}$. 
On the other hand, we have
\begin{eqnarray*}
        & \displaystyle\sum_{t=1}^{k-2} \sum_{i=t+1}^{2k-3} \frac{t \binom{2k-2t-1}{i-t+1}}{\binom{2k-1}{i}}& =\sum_{t=1}^{k-2} t \sum_{i=t+1}^{2k-3} \frac{\binom{2k-2t-1}{i-t+1}}{\binom{2k-1}{i}} \\
        && = \sum_{t=1}^{k-2} t \sum_{i=t+1}^{2k-3} \frac{i(i-1)\cdots (i-t+2)(2k-i-1)\cdots (2k-i-t-1)}{(2k-1)\cdots (2k-2t)} \\
        && \sim  \sum_{t=1}^{k-2} t \sum_{i=t+1}^{2k-3} \frac{i^{t-1} (2k-i)^{t+1}}{(2k)^{2t}}\\
        %&& = \sum_{t=1}^{k-2} \frac{t}{(2k)^{2t}} \sum_{i=t+1}^{2k-3} i^{t-1}\sum_{s=0}^{t+1}(-1)^s \binom{t+1}{s}(2k)^{t+1-s} i^s\\
        && = \sum_{t=1}^{k-2} \frac{t}{(2k)^{2t}} \sum_{s=0}^{t+1}(-1)^s \binom{t+1}{s} (2k)^{t+1-s} \sum_{i=t+1}^{2k-3} i^{t-1+s} \\
        && \sim \sum_{t=1}^{k-2} \frac{t}{(2k)^{2t}} \sum_{s=0}^{t+1}(-1)^s \binom{t+1}{s} (2k)^{t+1-s} \frac{(2k)^{t+s}}{t+s},
        %&& =2k \sum_{t=1}^{k-2} t \sum_{s=0}^{t+1}(-1)^s \binom{t+1}{s}  \frac{1}{t+s}\\
\end{eqnarray*}
where we used the standard Buchmann-Landau notation and all estimates are for $k \to \infty$.
Next, using the identity~\cite[Equation 5.41]{graham1994concrete}: 
\begin{equation}\label{eq:identity1}
    \displaystyle \frac{1}{\binom{z+w}{w}}= z \sum_{r=0}^{w}(-1)^r \binom{w}{r}  \frac{1}{z+r},
\end{equation}
the expression simplifies as
\begin{eqnarray*}
        \displaystyle\sum_{t=1}^{k-2} \sum_{i=t+1}^{2k-3} \frac{t \binom{2k-2t-1}{i-t+1}}{\binom{2k-1}{i}} = 2k \sum_{t=1}^{k-2} \frac{1}{\binom{2t+1}{t+1}}=k \sum_{t=1}^{k-2} \left(1+\frac{1}{2t+1}\right) \frac{1}{\binom{2t}{t}}.
    \end{eqnarray*}
    We finally use two standard combinatorial identities; see e.g.~\cite[Theorem 3.4]{Sprugnoli2006}):
    $$\displaystyle
     \sum_{t=0}^{\infty} \frac{1}{\binom{2t}{t}}= \frac{2\pi \sqrt{3}}{27}+\frac{4}{3},  \quad  \quad \sum_{t=0}^{\infty} \frac{1}{(2t+1)\binom{2t}{t}}=\frac{2\sqrt{3}\pi}{9},
    $$
    obtaining
    \begin{eqnarray*}
    \displaystyle
    \displaystyle \lim_{k \to \infty}\frac{1}{k}\sum_{t=1}^{k-2} \sum_{i=t+1}^{2k-3} \frac{t \binom{2k-2t-1}{i-t+1}}{\binom{2k-1}{i}}=\frac{2\pi \sqrt{3}}{27}+\frac{4}{3}+\frac{2\sqrt{3}\pi}{9}-2 =\frac{-18+8\sqrt{3}\pi}{27},
    \end{eqnarray*}
    which is the desired result. 
\end{proof}

We now turn to the second proof of Theorem~\ref{thm:conjproof}, which offers a straightforward argument that avoids using Theorem~\ref{thm:expC(2k,k)}.
Our approach uses some quantities introduced in~\cite{GruicaBarLevRavagnaniYaakobi2024} and defined as follows:
    \[\alpha_i(G,s)=|\{S \subseteq \{1, \ldots,n\} \,: \, |S|=s, \, e_i \in \langle G^j \mid  j \in S\rangle\}|\] 
%\rati{where for $1\le j\le n$, $G^j$ denotes the $j$-th column of $G$. }
The expectation $\mathbb{E}[\tau_i(G)]$ can be expressed in terms of the $\alpha_i(G,s)$'s, as the following result illustrates. 

\begin{proposition}[\text{\cite[Lemma 1]{GruicaBarLevRavagnaniYaakobi2024}}]\label{expected}
We have
\begin{equation*}
\mathbb{E}[\tau_i(G)]=\sum_{s=0}^{n-1} \frac{\binom{n}{s}-\alpha_i(G,s)}{\binom{n-1}{s}}.
\end{equation*}
\end{proposition}

The first step is precisely the computation of the $\alpha_i$'s.

\begin{lemma}\label{lem:conjalpha}
Let \( G_{k \times 2k} \) be the generator matrix of the \([2k,k]\) binary GRS code defined in Section~\ref{a particular code}. For any \( 1 \le i \le k \) and \( 0 \le s \le 2k-1 \), the number \( \alpha_i(G_{k \times 2k}, s) \) of recovery sets of size \( s \) for coordinate \( i \) is given by
\begin{multline*}
\alpha_i(G_{k \times 2k},s) = \binom{2k}{s} - \Biggl[ \binom{2k-1}{s} - 2\sum_{\ell=2}^{k} \binom{2k - 2\ell + 1}{s - \ell} + \\ \sum_{j=4}^{k+1}(j-3)\binom{2k - 2j + 3}{s - j} + (k-1)\delta_{s,k+1} \Biggr],
\end{multline*}
where \( \delta_{s,k+1} \) is the Kronecker delta, and $\binom{a}{b}=0$ whenever $b>a$.
\end{lemma}

\begin{proof}
    Fix \( i \in \{1,\dots,k\} \). By symmetry, we assume without loss of generality that $i = 1$.
    Define $M_{1,1}=\{e_1\}$, $M_{2,1}=\{e_1+e_2,e_2\}$, $M_{2,2}=\{e_k+e_1,e_k\}$, $M_{3,1}=\{e_1+e_2,e_2+e_3,e_3\}$, $M_{3,2}=\{e_k+e_1,e_{k-1}+e_k,e_{k-1}\}$, $\dots$, $M_{k,1}=\{e_1+e_2,e_2+e_3,\dots,e_{k-1}+e_k,e_k\}$, $M_{k,2}=\{e_2+e_3,e_3+e_4,\dots,e_{k}+e_1,e_2\}$. Note that a set $S$ is a recovery set for $e_1$ if and only if it contains one of the sets $M_{\ell, j}$ for $1\le \ell\le k$ and $1\le j\le 2$, and that such sets are minimal with this property. 
%Define $s_1=1$ and, for any $i\in\{(2,1),(2,2), \dots, (k-1,1),(k-1,2)\}$, $s_i=s_{(l,m)}=l$.

Remember that $\alpha_1(G_{k \times 2k},s)$ is the number of all the recovery sets of $e_1$ of size $s$. Our strategy is to compute, for any $\ell\in\{1, \dots, k\}$, the number of such recovery sets containing a minimal set of size $\ell$ and no smaller recovery set.
The claim will then follow by summing the values over $\ell$.

For $\ell=1$, the desired number is~$\smash{\binom{2k-1}{s-1}}$. For $2\le \ell\le k$, the desired number is 
     \begin{equation}\label{eq:valuesforell}
        2\binom{2k-2\ell+1}{s-\ell}-\sum_{r=2}^{\min \{\ell,k+1-\ell\}}(2-\delta_{\ell r}) \binom{2k-2\ell-2r+3}{s-(\ell+r)}-(2-\delta_{k-\ell+2,\ell})\delta_{s,k+1}
     \end{equation}
    where $\delta_{ab}$ is the Kronecker delta evaluated at the pair $(a,b)$. 
    For any $2\le \ell\le k$, consider the two minimal sets of length $\ell$, $M_{\ell,1}$, and $M_{\ell,2}$. 
    
    The number of recovery sets containing $M_{\ell,1}$ (resp. $M_{\ell,2}$) and not $M_{t,1}$ (resp. $M_{t,2}$) for any $t< \ell$ is \begin{equation}\label{eq:term1}
        \binom{2k-\ell-(\ell-1)}{s-\ell}.
    \end{equation}
    Each such recovery set is obtained by fixing the $\ell$ elements of $M_{\ell,1}$ (resp. $M_{\ell,2}$), and choosing other $s-\ell$ elements among the others, where we exclude the $\ell$ elements of~$M_{\ell,1}$ (resp. $M_{\ell,2}$) and the terms $e_1,e_2,\dots,e_{\ell-1}$ (resp. the terms $e_1,e_k,e_{k-1},\dots,e_{k-(\ell-3)}$), for the constraint on~$M_{t,1}$ (resp. $M_{t,2}$) with $t< \ell$. 
    From this number, we must subtract the number of recovery sets containing a smaller recovery set.
    Since the recovery set $S$ containing $M_{\ell,1}$ (resp. $M_{\ell,2}$) does not contain $M_{t,1}$ (resp. $M_{t,2}$) with $t< \ell$, this can only happen if it contains $M_{r,2}$ (resp. $M_{r,1}$) for some $r< \ell$.
    We now count, for each $r<\ell$, the number of recovery sets containing $M_{\ell,1}\cup M_{r,2}$ (resp. $M_{\ell,2}\cup M_{r,1}$) and not $M_{u,1}$ and $M_{t,2}$ (resp. $M_{u,2}$ and $M_{t,1}$) for any $u<\ell$ and $t<r$.
    Notice that $r\leq\ell\leq k$.
    Assume first $s\neq k+1$.
    
    \textbf{Case 1.} $\ell+r\le k+1$. In this case the sets $M_{\ell,1}$ and  $M_{r,2}$ (resp. $M_{\ell,2}$ and  $M_{r,1}$) are disjoint.
    The sought number is the number of choices of $s-\ell-r$ elements among the remaining allowed elements, namely excluding the elements of $M_{\ell,1}$, of $M_{r,2}$ and $e_1,\ldots,e_{\ell-1},e_k,\ldots,e_{k-(r-3)}$ (resp. $e_1,e_k,e_{k-1},\ldots,e_{k-(\ell-3)}, e_2,e_3,\ldots,e_{r-1}$). This is \begin{equation}\label{eq:term2}
        \binom{2k-\ell-(\ell-1)-r-(r-1)+1}{s-\ell-r}.
    \end{equation} 
    Notice that when $r=\ell$, we consider the recovery sets containing $M_{\ell,1}\cup M_{\ell,2}$. This is included both in the count of $M_{\ell,1}$ and in that of $M_{\ell,2}$, and therefore must be subtracted only once, explaining the Kronecker delta $\delta_{\ell,r}$.
    
    \textbf{Case 2.} $\ell+r \geq k+2$. In this case, the term $e_{k-r+2}$ (resp. the term $e_r$) belongs to $M_{r,2}$ and $M_{\ell,1}$ (resp. $M_{r,1}$ and $M_{\ell,2}$). 
    If the inequality is strict, then no such set was counted in~\eqref{eq:term1}, because $k-r+2<\ell$ and therefore the term $e_{k-r+2}$ (resp. the term $e_r$) was among the excluded.
    If equality holds, then $e_{k-r+2}=e_\ell$ (resp. $e_r=e_{k-\ell+2}$), and it is their only common term. 
    If it holds also $s=k+1$, then precisely one such set was counted in~\eqref{eq:term1}, namely the set $M_{r,2}\cup M_{\ell,1}$ (resp. $M_{r,1}\cup M_{\ell,2}$). Again, note that when $r=\ell$, this set must be counted, but we should subtract it once in total to account for the repetition; hence the factor $(2-\delta_{k+2-\ell,\ell})$. If $s\geq k+2$ then no such set was counted, because this would force the set to contain one of the excluded elements. This explains the number~\eqref{eq:valuesforell}.
    Summing up we obtain that $\alpha_1(G_{k \times 2k},s)$ is given by
    \[\binom{2k-1}{s-1}+\sum_{\ell=2}^k \left(2\binom{2k-2\ell+1}{s-\ell}-\hspace{-3.1mm}\sum_{r=2}^{\min \{\ell,k+1-\ell\}}\hspace{-4.3mm}(2-\delta_{\ell r}) \binom{2k-2\ell-2r+3}{s-(\ell+r)}-(2-\delta_{k-\ell+2,\ell})\delta_{s,k+1}\right)\nonumber.\]
    Changing the double summation to a single one with parameter $j=r+\ell$, this reads  
        \[\binom{2k-1}{s-1}+\sum_{\ell=2}^k 2\binom{2k-2\ell+1}{s-\ell}-\sum_{j=4}^{k+1} (j-3) \binom{2k-2j+3}{s-j}+(k-1)\delta_{s,k+1};\label{eqn:3}\]
    which coincides with the claim.
\end{proof}

We conclude this section with the second proof of  Theorem~\ref{thm:conjproof}.

\begin{proof}[Proof of Theorem~\ref{thm:conjproof}]
By \cite[Lemma 1]{GruicaBarLevRavagnaniYaakobi2024}, we know that 
\begin{equation*}
\mathbb{E}[\tau_i(G)]=\sum_{s=0}^{n-1} \frac{\binom{n}{s}-\alpha_i(G,s)}{\binom{n-1}{s}}.
\end{equation*}
By substituting the value of $\alpha_i(G,s)$ from Lemma~\ref{lem:conjalpha}, we obtain 
\begin{eqnarray}
    \mathbb{E}[\tau_i(G)]%=\sum_{s=0}^{2k-1} \frac{\binom{2k-1}{s}-\sum_{l=2}^k 2\binom{2k-2l+1}{s-l}+\sum_{i=4}^{2k}  (i-3) \binom{2k-2i+3}{s-i}}{\binom{2k-1}{s}} \nonumber\\
    =2k-2 \sum_{s=0}^{2k-1} \sum_{\ell=2}^k \frac{\binom{2k-2\ell+1}{s-\ell}}{\binom{2k-1}{s}}+\sum_{s=0}^{2k-1} \sum_{j=4}^{k+1} (j-3) \frac{\binom{2k-2j+3}{s-j}}{\binom{2k-1}{s}}+\frac{k-1}{\binom{2k-1}{k+1}}. \label{eqn:1}
\end{eqnarray}
Note that the $\lim_{k\mapsto\infty}\frac{k-1}{\binom{2k-1}{k+1}}\frac{1}{k}=0$. Since our goal is to compute $\lim_{k \to \infty} \frac{\mathbb{E}[\tau_i (G_{k\times2k})]}{k}$, we now consider only the other terms.
Define 
$$ I_1 = \sum_{s=0}^{2k-1} \sum_{\ell=2}^k \frac{\binom{2k-2\ell+1}{s-\ell}}{\binom{2k-1}{s}}, \qquad  I_2 = \sum_{s=0}^{2k-1} \sum_{j=4}^{k+1} (j-3) \frac{\binom{2k-2j+3}{s-j}}{\binom{2k-1}{s}}. $$\\
To calculate $I_1$ and $I_2$, we use the same asymptotic formulas %~\eqref{eq:asymptote1} \red{refs!} 
and proceed with similar steps as in the first proof of Theorem~\ref{thm:conjproof}. Thus,
\begin{align*}
    I_1& =\sum_{s=0}^{2k-1} \sum_{\ell=2}^k \frac{s(s-1)\cdots (s-\ell+1)(2k-s-1)\cdots (2k-s-\ell+2)}{(2k-1)\cdots (2k-2\ell+2)}\\
    & \sim \sum_{s=0}^{2k-1} \sum_{\ell=2}^k \frac{s^\ell(2k-s)^{\ell-2}}{(2k)^{2\ell-2}}
    =\sum_{\ell=2}^k \sum_{r=0}^{\ell-2} (-1)^r \binom{\ell-2}{r} \frac{1}{(2k)^{\ell+r}}\sum_{s=0}^{2k-1} s^{\ell+t}\\
    & \sim \sum_{\ell=2}^k \sum_{r=0}^{\ell-2} (-1)^r \binom{\ell-2}{r} \frac{1}{(2k)^{\ell+r}} \frac{(2k)^{\ell+r+1}}{\ell+r+1}\\
    &=2k \sum_{\ell=2}^k \sum_{r=0}^{\ell-2} (-1)^r \binom{\ell-2}{r}  \frac{1}{\ell+r+1}
\end{align*}
By the identity in \eqref{eq:identity1}, this reduces to
\begin{align*}
    I_1& =2k \sum_{\ell=2}^k \frac{1}{\ell+1} \frac{1}{\binom{2\ell-1}{\ell+1}}
    =2k \left(\sum_{\ell=1}^k \frac{1}{\ell} \frac{1}{\binom{2\ell}{\ell}}-\sum_{\ell=1}^k \frac{1}{(2\ell+1)} \frac{1}{\binom{2\ell}{\ell}}\right).
    \end{align*}

We now turn to $I_2$, computing 
 \begin{eqnarray*}
    &I_2&
    =\sum_{s=0}^{2k-1} \sum_{j=4}^{2k} (j-3) \frac{s(s-1)\cdots (s-j+1)(2k-s-1)\cdots (2k-s-j+4)}{(2k-1)\cdots (2k-2j+4)}\\
    && \sim \sum_{s=0}^{2k-1} \sum_{j=4}^{k+1} (j-3) \frac{s^j (2k-s)^{j-4}}{(2k)^{2j-4}}
    =  \sum_{j=4}^{k+1} (j-3)  \sum_{r=0}^{j-4} (-1)^r \binom{j-4}{r} \frac{1}{(2k)^{j+r}} \sum_{s=0}^{2k-1} s^{j+r}\\
    &&\sim \sum_{j=4}^{k+1} (j-3) \sum_{r=0}^{j-4} (-1)^r \binom{j-4}{r} \frac{1}{(2k)^{j+r}}  \frac{(2k)^{j+r+1}}{j+r+1}\\
    && =2k \sum_{j=4}^{k+1} (j-3) \sum_{r=0}^{j-4} (-1)^r \binom{j-4}{r} \frac{1}{j+r+1}.
    \end{eqnarray*}
Again by using the identity in~\eqref{eq:identity1}, we have 
    \begin{align*}
    I_2&=2k\sum_{j=4}^{k+1} \frac{j-3}{j+1} \frac{1}{\binom{2j-3}{j-4}}\\
    &= 2k\left(\frac{1}{2}\sum_{j=0}^{k-1} \frac{1}{\binom{2j}{j}}-\frac{3}{2}\sum_{j=0}^{k-1}\frac{1}{(2j+1)}\frac{1}{\binom{2j}{j}}+2\sum_{j=1}^{k-1} \frac{1}{j} \frac{1}{\binom{2j}{j}}\right).
\end{align*}
Substituting the values of $I_1$ and $I_2$ in \eqref{eqn:1} and using the asymptotic identities 
from \cite[Theorem 3.4]{Sprugnoli2006}, we obtain
$$
\lim_{k\rightarrow \infty}\frac{\mathbb{E}[\tau_i(G)]}{k}
=\frac{-18+8\sqrt{3}\pi}{27},
$$
which completes the proof.
\end{proof}

We point out that the second proof based on the $\alpha_i(G,s)$'s has two advantages.
The first one is that it avoids  solving the recurrence for $B(k,i)$ 
and shortens the argument.
The second one is that once the $\alpha_i(G,s)$'s are known, the higher moments and the probability mass function of~$\tau_i(G)$ can be computed as well, as we will show in Section~\ref{higher}.

%%%%%%%%%%%%%%%%
%%%%%%%%%%%%%%%%
\section{Expectation of codes from balanced quasi-arcs}\label{quasi-arcs}
In this section, we investigate the asymptotic behaviour of another family of codes, introduced in~\cite{GruicaMontanucciZullo2024} via geometric objects called \textit{balanced quasi-arcs}. These codes all have dimension $k=3$ and depend on the choice of integers $x$ and $y$. The code length is
$n=3x+3y$. We refer to~\cite{GruicaMontanucciZullo2024} for the details of the construction (as we don't need them here), but we briefly survey the contributions made by that paper and by~\cite{boruchovsky2025makingfirstrandomaccess}.
In~\cite{GruicaMontanucciZullo2024}, the authors derive a formula for the expected value for arbitrary $x$ and $y$. 
They then compute an upper bound for the limit superior of the expected value $\mathbb{E}(\tau_i(G_{x,x}))$ as $x$ goes to infinity. Moreover, \cite[Remark 4.9]{GruicaMontanucciZullo2024} mentions that 
computational results indicate that the ratio $y/x$ offering the best performance asymptotically
is close to $0.85$. 
In~\cite{boruchovsky2025makingfirstrandomaccess}, 
an upper bound for the asymptotic performance is computed for any $x$, $y$ with given ratio $y/x=\varepsilon$.
In this paper, we compute the exact asymptotic performance of the same codes for any $x$, $y$ with given ratio $y/x=\varepsilon$. Moreover, we prove that the best performance is obtained for $\varepsilon\approx 0.833968$. Our approach is based on the Dominated Convergence Theorem.

Our starting point is the following formula for the expectation $\mathbb{E}[\tau_i(G_{x,y})]$.

\begin{proposition}[\text{\cite[Corollary 4.7]{GruicaMontanucciZullo2024}}]\label{25.02.24} For all $i$ we have
\begin{eqnarray*}
&\mathbb{E}[\tau_i(G_{x,y})] = &3 +  
\frac{2}{3x + 3y - 2}
- \frac{y - 1}{3x + 3y - 1}
- \displaystyle \frac{2(xy + \binom{x}{2}) + y(3x + 2y) + \frac{y(y-1)}{2}}{\binom{3x + 3y - 1}{2}} +
\\ &&
\sum_{s=3}^{x+2y}
\prod_{j=0}^{s-1}
\frac{x + 2y - j}{3x + 3y - j - 1}
+ \sum_{s=3}^{y+1}
\frac{2\binom{y}{s-1}x}{\binom{3x+3y-1}{s}}.
\end{eqnarray*}
\end{proposition}

Codes of the form $G_{x,y}$ outperform previously known constructions for achieving a low expectation~\cite{GruicaMontanucciZullo2024}. In the case where $x=y$, one has  
$ \lim_{x\mapsto\infty}\mathbb{E}[\tau_i(G_{x,x})]\leq 0.88\bar{2}k$;
see~\cite[Theorem~4.8]{GruicaMontanucciZullo2024}.

The first result of this section is the following.

\begin{theorem}\label{25.02.25}
Assume that $y$ grows linearly with $x$, i.e., let $\varepsilon=y/x$ be constant. For all $i$, we have
\begin{align*}
\lim_{x\to\infty}\mathbb{E}[\tau_i(G_{x,y})]=\frac{153 + 543 \varepsilon + 805 \varepsilon^2 + 611 \varepsilon^3 + 234 \varepsilon^4 + 
 36 \varepsilon^5}{3 (1 + \varepsilon)^2 (2 + \varepsilon) (3 + 2 \varepsilon)^2}.
\end{align*}
\end{theorem}
\begin{proof}
We use Proposition~\ref{25.02.24}, and we compute the limit of each term. The only non-trivial computations are those for
\[
L_1(\varepsilon x)=\lim_{x\to\infty}\sum_{s=3}^{x+2\varepsilon x}
\prod_{j=0}^{s-1}
\frac{x + 2\varepsilon x - j}{3x + 3\varepsilon x - j - 1} \quad \text{and}\quad L_2(\varepsilon)=\lim_{x\to\infty}\sum_{s=3}^{\varepsilon x+1}
\frac{2\binom{\varepsilon x}{s-1}x}{\binom{3x+3\varepsilon x-1}{s}}.
\]
We let
$$f(x,s)=\smash{\prod_{j=0}^{s-1} 
\frac{x(2\varepsilon+1) - j}{x(3\varepsilon+3) - j - 1}}, \qquad  g(x,s)=\frac{2\binom{\varepsilon x}{s-1}x}{\binom{x(3\varepsilon+3)-1}{s}}$$ and apply the dominated convergence theorem, after checking its applicability.
Note that for $\varepsilon,j\geq0$ and for any $x\geq 1$ we have
$\smash{\frac{x(2\varepsilon+1) - j}{x(3\varepsilon+3) - j - 1}\leq \frac{2}{3}}$, from which $\smash{f(x,s)\leq \left(\frac{2}{3}\right)^s}$. 
Since $\smash{\sum_{3}^\infty \left(\frac{2}{3}\right)^s< \infty}$, we can indeed apply the dominated convergence theorem, obtaining
\[
L_1(\varepsilon)=
\lim_{x\to\infty}\sum_{s=3}^{x(2\varepsilon+1)}
f(x,s)=\sum_{s=3}^{\infty}
\lim_{x\to\infty}f(x,s)=\sum_{3}^\infty \left(\frac{2\varepsilon+1}{3\varepsilon+3}\right)^s=\frac{\left(\frac{2\varepsilon+1}{3\varepsilon+3}\right)^3}{1-\left(\frac{2\varepsilon+1}{3\varepsilon+3}\right)}.
\]

We now turn to $L_2(\varepsilon)$.
Using the following inequalities obtained from the Stirling series, see for example \cite[pag.262]{WhittakerWatson2021}:
\[
\frac{n^k}{k^k} \leq \binom{n}{k} \leq \frac{n^k}{k!} < \left(\frac{n  e}{k}\right)^k,
\qquad \sqrt{2\pi n} \left(\frac{n}{e}\right)^n e^{\left(\frac{1}{12n} - \frac{1}{360n^3}\right)} < n! < \sqrt{2\pi n} \left(\frac{n}{e}\right)^n e^{\frac{1}{12n}},
\] 
it can be shown with straightforward computations that 
\begin{align*}
\displaystyle
g(x, s) \leq \frac{2}{((3k+3)-1)} 
\frac{e^{\log(s)+(s-1)\left(\log(\frac{s}{s-1}\, \frac{\varepsilon  e }{3\varepsilon +3-1})\right)-\frac{1}{12(s-1)} +\frac{1}{360(s-1)^3}}}{\sqrt{2\pi (s-1)}}.
\end{align*}
For large enough $s$ we have $\smash{{\frac{s}{s-1}\frac{\varepsilon e}{3\varepsilon +3-1} < 1}}$, and therefore $g(x, s)$ is bounded from above by a function of the form $c \sqrt{s}e^{-s}$, which is integrable in the interval $[3,\infty)$.
By the Dominated Convergence Theorem, we obtain
\begin{align*}\displaystyle
L_2(\varepsilon)&=\lim_{x\to\infty}\sum_{s=3}^{\varepsilon x+1}f(x,s)=\sum_{s=3}^{\infty}
\lim_{x\to\infty}f(x,s)\\
&=\frac{2}{3\varepsilon+3}\sum_{3}^\infty 
 s\left(\frac{\varepsilon}{3\varepsilon+3}\right)^{s-1}
=\, \frac{2}{3\varepsilon+3}\left(-\frac{\varepsilon^2 (-3 + 2 \varepsilon)}{(-1 + \varepsilon)^2}\right),
\end{align*}
which completes the proof.
\end{proof}
For $y=0$ (and thus $\varepsilon=0$), Theorem~\ref{25.02.25}
recovers~\cite[Theorem 4.5]{GruicaMontanucciZullo2024}.
In fact, it shows that~\cite[Theorem 4.5]{GruicaMontanucciZullo2024} is sharp.
For $x=y$, Theorem~\ref{25.02.25} implies
\[
\lim_{x\mapsto\infty}\mathbb{E}[\tau_i(G_{x,y})]= \frac{397}{150}2.62381 \approx 0.88\bar{2}k,
\]
which is precisely the upper bound $0.88\bar{2}k$ from \cite[Theorem 4.8]{GruicaMontanucciZullo2024}.

As a corollary of Theorem~\ref{25.02.25}, we can easily find the optimal ratio $\varepsilon=y/x$ that minimizes the expectation.

\begin{corollary}
The optimal ratio $\varepsilon=y/x$ that minimizes the expectation $\mathbb{E}[\tau_i(G_{x,y})]$ for large~$x$ (and thus $y$) is $\varepsilon\approx 0.833968$. For such value,
\[\lim_{x\mapsto\infty}\mathbb{E}[\tau_i(G_{x,y})]\approx  0.881542k.\]
\end{corollary}

\begin{proof}
Using Theorem~\ref{25.02.25}, we look at $\lim_{x\to\infty}\mathbb{E}[\tau_i(G_{x,y})]$ as a function of $\varepsilon$. The derivative is \[
\frac{-261 - 513 \varepsilon - 63 \varepsilon^2 + 583 \varepsilon^3 + 592 \varepsilon^4 + 238 \varepsilon^5 + 
 36 \varepsilon^6}{3 (1 + \varepsilon)^3 (2 + \varepsilon)^2 (3 + 2 \varepsilon)^3}.
\]
By Descartes’ Rule of Signs, the number of positive real roots of the polynomial at the numerator is at most one, and a numerical computation confirms the existence of exactly one positive real root.
Moreover, it can be checked that
the polynomial has a non-solvable Galois group.
Since by the Abel-Ruffini Theorem there is no general method for solving degree-6 polynomial with non-solvable Galois group, we only include a numerical approximation of the root, $\varepsilon\approx 0.833968$, corresponding to the minimum
$\approx 0.881542k$.
\end{proof}

%%%%%%%%%%%%%%%%%%%%%%%%%%%%%%%%%%%%%%%%%%%%%%%%%%%%%%%%%%%%%%%%%%%%%%%%%%%%%%%%

\section{Higher moments and applications}\label{higher}

In~\cite{GruicaBarLevRavagnaniYaakobi2024}, it was shown that many classical families of error-correcting codes all have the same expectation $k$, i.e., the dimension of the code.
This value of the expectation is implied by the symmetry of classical code constructions, including MDS, simplex, and Hamming codes.

It is therefore natural to ask which properties differentiate the performance of these codes. A natural direction we take in this paper is to compute the higher moments of the random variable $\tau_i(G)$. 
Our first result shows that the higher moments of $\tau_i(G)$ are completely determined by the quantities $\alpha_i(G,s)$.
However, proving this fact is not as easy as proving Proposition~\ref{expected}.

We start by introducing the 
 ``ordered'' version of the quantities $\alpha_i(G,s)$, counting the ordered lists of size $s$ containing a recovery set for the $i$-th information strand. In symbols, we let
\[\beta_i(G,r)=|\{(\rho_1,\ldots,\rho_r) \in \{1,\dots,n\}^r: e_i \in \langle G^{\rho_1},\ldots, G^{\rho_r} \rangle\ \}|. \]
These quantities will be particularly convenient in our proofs, and they carry the same information as the $\alpha_i(G,s)$'s. The following inversion formulas make this statement precise.

\begin{lemma}[$\alpha\beta$-inversions]\label{alphabeta inversions}
The quantities $\myalpha_i(G,r)$ and $\beta_i(G,r)$ satisfy the following inversion formulas:  
\begin{enumerate}
   \item $\displaystyle \beta_i(G,r) = \sum\limits_{s=1}^r \genfrac\{\}{0pt}{}{r}{s} \, s! \, \alpha_i (G,s)$,
    \item $\displaystyle \alpha_i(G,s) = \frac{1}{s!} \sum\limits_{r=1}^s \genfrac[]{0pt}{}{s}{r} (-1)^{s-r} \beta_i(G,r)$,
\end{enumerate}
where ${ s \brack r}$ and ${ s \brace r}$ are the Stirling numbers of first and second kinds, respectively.
\end{lemma}

\begin{proof}
Fix $r$ and for a set $S\subseteq \{1,\dots,n\}$ consider the collection $A_S$ of sequences in $\{1,\dots,n\}^r$ with underlying set $S$, namely
\[
A_S = \left\{(\rho_1, \dots, \rho_r) \in \{1,\dots,n\}^r \mid  \{\rho_1, \dots, \rho_r\} = S \right\}.
\]
If $|S| = s$, we have $\smash[b]{\displaystyle |A_S| = { r \brace s} s!}$ by the definition of Stirling numbers. We also have
\[
\{1,\dots,n\}^r = \bigcup_{s=1}^r \bigcup_{\substack{S\subseteq \{1, \ldots,n\}\\ |S|=s}} A_S,
\]
where the unions are disjoint.
By definition, $\myalpha_i(G,s)$ counts the subsets $S$ with $e_i \in \langle G^j : j \in S \rangle$. Since $A_S \subseteq \mybeta_i(G,r)$ if and only if this holds, we obtain  
\[
\mybeta_i(G,r) = \sum\limits_{s=1}^r { r \brace s} s! \, \myalpha_i (G,s),
\]
which establishes Part 1. 

Part 2 follows from the well-known inversion formulas for the Stirling numbers of the first and second kind.
\end{proof}

While the $\alpha_i(G,s)$'s and the  $\beta_i(G,r)$'s carry equivalent information, we find the 
quantities~$\beta_i(G,r)$ more naturally linked to
the probability distribution of $\tau_i(G)$, in the following precise sense:
\begin{equation}\label{obs}
\mathbb{P}(\tau_i(G)>r)=\frac{n^r-\beta_i(G,r)}{n^r}.    
\end{equation}

\begin{remark}
Since $\alpha_i(G,s)$ is only defined for $s\leq n$, one would expect a finite number of the $\beta_i(G,r)$'s to determine all the others.
This is indeed the case, as one can show using the inversion formulas for the Stirling numbers of the first and second kind:
\begin{align*}
\mybeta_i(G,n+r)&=\sum_{s=0}^n {n+r \brace s} s! \myalpha_i(G,s)
=\sum_{s=0}^n{n+r \brace s}\sum_{j=0}^s {s \brack j} (-1)^{s-j} \mybeta_i(G,r). 
\end{align*}
\end{remark}

We present a closed formula for $\mathbb{E}(\tau_i(G)^p)$ in terms of $\alpha_i(G,s)$. Later, we will provide an explicit formula for the probability mass function of $\tau_i(G)$; see Theorem \ref{thm: probability density function}. 
We find it convenient to work with the following quantities: 
\begin{align*}
\myalphatilde_i(G,s) &= \binom{n}{s}-\myalpha_i(G,s), \qquad
\mybetatilde_i(G,r) = n^r-\mybeta_i(G,r) .
\end{align*}

\begin{theorem}[Higher moments]\label{p-th moment general formula}
The $p$-th moment of the random variable $\mytau_i(G)$ is
\begin{align*}
\mathbb{E}[\tau_i(G)^p]=\displaystyle\sum_{s=0}^{n-1}  \myalphatilde_i(G,s)s!\sum_{l=0}^{p-1}\binom{p}{l} \left[\sum_{b=0}^l \left(\frac{1}{n}\right)^{l}{l \brace b }
D^{(b)}(f_1\cdots f_{s+1})(1/n)
\right],
\end{align*}
where $\displaystyle f_j(x) = 1/(1-jx)$ for $1\leq j \leq s$, $\displaystyle f_{s+1}(x) = x^s$, and 
$D^{(b)}(f_1\cdots f_{s+1})(1/n)$
is the evaluation of the $b$-th derivative of the product $f_1 \cdots f_{s+1}$ at $1/n$. The latter is given by
$$\displaystyle D^{(b)}(f_1\cdots f_{s+1})(1/n)=\sum_{b_1+\cdots+b_{s+1}=b}\binom{b}{b_1,\dots ,b_{s+1}}\prod f_j^{(b_j)}(1/n),$$ 
with $$f_j^{(u)}(x)=\frac{u!j^u}{(1-jx)^{u+1}} \mbox{ for $j\in\{1,\dots,s\}$}.$$
\end{theorem}
\begin{proof}
We consider the tail-sum formula for the $p$-th moment and combine it with~\eqref{obs}, obtaining 
\[\mathbb{E}[\tau_i(G)^p]= \sum_{r=0}^{+\infty}((r+1)^p-r^p)\, \frac{n^r-\beta_i(G,r)}{n^r}.\]
Applying the $\alpha\beta$-inversions of Proposition~\ref{alphabeta inversions} and changing the order of summation, we obtain
\begin{eqnarray*} \mathbb{E}[\tau_i(G)^p]= \displaystyle\sum_{r=0}^{+\infty}\sum_{s=0}^{n-1} s!((r+1)^p-r^p)\left(\frac{1}{n^r}\right){r \brace s}\myalphatilde_i(G,s).
\end{eqnarray*}
Note that this is a series with positive terms.
Expanding the power $(r+1)^p$ with the Binomial Theorem and changing again the order of summation, we get 
\begin{eqnarray}\label{step1}
\mathbb{E}[\tau_i(G)^p]= \displaystyle\sum_{s=0}^{n-1}  \myalphatilde_i(G,s)s!\sum_{l=0}^{p-1}\binom{p}{l}\sum_{r=0}^{+\infty} {r \brace s} r^l\frac{1}{n^r}.
\end{eqnarray}
We now focus on the innermost summation, which we denote by $C_l$. We have $C_l = \tilde{C}_l(1/n)$, where $\smash{\tilde{C}_l(x) = \sum_{r=s}^{+\infty} {r \brace s} r^lx^r}$. 
Note that the evaluation at $1/n$ is well defined, since for $0\le s\le n-1$ and $x\geq 0$ we have $\smash{0 < \tilde{C}_l(x) < \sum_{r=0}^{+\infty} \frac{r^l}{s!}\left(sx\right)^r}$. By the ratio test, this implies that
$\tilde{C}_l(x)$ and all its derivatives have a convergence radius of at least $1/n$.

Consider the rational generating function of the Stirling numbers of the second kind: \begin{equation}\label{eq from infinite to finite product}
\sum_{r=s}^{+\infty} {r \brace s}  x^{r}=\left(\prod_{j=1}^s \frac{1}{1-jx}\right) x^s.
\end{equation}
We wish to take the $b$-th derivative of both sides. 
For the right-hand side
define $ f_j = \frac{1}{1-jx}$ for $1\le j\le s$ and $f_{s+1}=x^s$. Then, by the Leibniz formula, we have
\[
D^{(b)}(f_1\cdots f_{s+1})=\sum_{b_1+\cdots+b_{s+1}=b}\binom{b}{b_1,\dots ,b_{s+1}}\prod f_j^{(b_j)},
\]
with $\smash{f_j^{(u)}(x)}$ as in the statement.

The $b$-th derivate of the left-hand side term is instead given by $$\frac{1}{x^b} \tilde{C}_{\underline{b}}(x) = \sum_{r=k}^{+\infty} {r \brace s} r^{\underline{b}}x^{r-b},$$ where 
$\smash{r^{\underline{b}} = \prod_{i=0}^{b-1}\left(r-i\right)}$ is the $b$-th falling factorial. 
Using the identity
$\smash{r^l = \sum_{j=0}^{l} {l \brace j }r^{\underline{j}}}$ we then obtain
$\smash{ \tilde{C}_l = \sum_{j=0}^{l} {l \brace j}\tilde{C}_{\underline{j}}}$. 
Putting the pieces together and setting $C_{\underline{b}} = \tilde{C}_{\underline{b}}\left(1/n\right)$ one gets
\begin{align*}
\mathbb{E}[\tau_i(G)^p]&=\displaystyle\sum_{s=0}^{n-1}  \myalphatilde_i(G,s)s!\sum_{l=0}^{p-1}\binom{p}{l}C_l=\\
&=
\displaystyle\sum_{s=0}^{n-1}  \myalphatilde_i(G,s)s!\sum_{l=0}^{p-1}\binom{p}{l} \left(\sum_{b=0}^l \left(\frac{1}{n}\right)^{b} {l \brace b } C_{\underline{b}}\right)\\
&=
\displaystyle\sum_{s=0}^{n-1}  \myalphatilde_i(G,s)s!\sum_{l=0}^{p-1}\binom{p}{l} \left(\sum_{b=0}^l \left(\frac{1}{n}\right)^{b}{l \brace b }
D^{(b)}(f_1\cdots f_{s+1})(1/n) \right).
\end{align*}
This completes the proof.
\end{proof}

As a corollary of Theorem~\ref{p-th moment general formula}, we obtain the following formula for the second moment and the variance. 
\begin{corollary}\label{second-moment}\label{cor: variance}
The second moment of the random variable $\mytau_i(G)$ is given by
\[
\mathbb{E}[\mytau_i(G)^2]
 =\displaystyle\sum_{s=0}^{n-1} \left[\frac{ \myalphatilde_i(G,s)}{\binom{n-1}{s}} \left(1+ 2{s}+ 2 \sum_{\ell=1}^s \frac{\ell}{(n-\ell)}\right)\right],
\]
and its variance by
\begin{eqnarray*}
\text{Var}[\mytau_i(G)]=
 \sum_{v=0}^{n-1} \left(\left(\frac{ \myalphatilde_i(G,v)}{\binom{n-1}{v}}\right)\left(2{v}+ 2 \sum_{\ell=1}^v \frac{\ell}{(n-\ell)}-\sum_{s=1}^{n-1} \frac{ \myalphatilde_i(G,s)}{\binom{n-1}{s}}\right)\right).
\end{eqnarray*}
\end{corollary}

Theorem~\ref{p-th moment general formula} shows that all the moments of $\tau_i(G)$ are completely determined by the $\alpha_i(G,s)$'s.
We can prove an even stronger result: the probability mass function itself is completely determined by the $\alpha_i(G,s)$'s.

\begin{theorem}\label{thm: probability density function}
The probability mass function of the random variable $\tau_i(G)$ is given by 
\begin{align*}
\mathbb{P}[\tau_i(G)=r]=\frac{1}{n^r}\sum \limits_{s=1}^{r} \left[ {r \brace s}-n {r-1 \brace s } \right]s! \alpha_i(G,s).
\end{align*}
\end{theorem}
\begin{proof}
Since $\mathbb{P}[\tau_i(G)=r]=\mathbb{P}[\tau_i(G)>r-1]-\mathbb{P}[\tau_i(G)>r]$, the claim follows immediately from formula \eqref{obs} and the $\alpha\beta$-inversions of Proposition \ref{alphabeta inversions}. 
\end{proof}

\section{MDS, Hamming and Simplex codes}\label{famous codes}

In this section, we apply our results to MDS, Hamming, and simplex codes. 
For MDS codes, we find an easy formula for the variance of $\tau_i(G)$ in terms of the harmonic numbers. 
We give an explicit formula for the probability mass function for each of the three families and conclude the section by comparing the probability mass functions for these families of codes. While these codes are \textit{recovery balanced} in the sense of~\cite[Definition 2]{GruicaBarLevRavagnaniYaakobi2024} and are therefore undistinguishable from the point of view of the expectation, the mass functions they induce are different.

\subsection{MDS codes} 
Suppose that $G$ is in systematic form and generates an $[n,k]_q$ MDS code.
In \cite[Corollary 1]{GruicaBarLevRavagnaniYaakobi2024}, the values $\alpha_i(G,s)$ have been computed as
\begin{equation}\label{eq:alphaMDS}
    \myalpha_i(G,s) = 
\begin{cases}
      \displaystyle \binom{n-1}{s-1} & \text{if}\ s < k, \\
      \displaystyle \binom{n}{s} & \text{if}\ s \geq k.
\end{cases}
\end{equation}

We provide a compact formula for the variance of systematic MDS codes.
\begin{corollary}
Let $G$ be the generator matrix in systematic form of an MDS $[n,k]_q$ code. The variance of the random variable $\mytau_i(G)$ is given by
\begin{equation}\label{eq:MDS variance}
\text{Var}[\mytau_i(G)]= 2n(n-k) (H_{n-k-1}-H_{n-1}) +2nk-(k-1)k;
\end{equation}
where $H_n$ is the $n$-th harmonic number.
\end{corollary}
\begin{proof}
By Equation~\eqref{eq:alphaMDS} and Corollary~\ref{cor: variance}, we have
\[
\text{Var}[\tau_i(G)]= 2 \sum_{v=0}^{k-1} \sum_{\ell=n-v}^{n-1} \left(\frac{n}{\ell}-1\right).
\]
Noticing that $\sum_{\ell=n-v}^{n-1} \frac{1}{\ell}=H_{n-1}-H_{n-v-1}$, the proof follows from a straightforward computation and using the identity 
$ \sum _{k=1}^{n}H_{k}=(n+1)H_{n}-n$.
\end{proof}

\begin{proposition}\label{MDS_density}
Let $G$ be the generator matrix in  systematic form of an MDS $[n,k]_q$ code. The probability mass function of the random variable $\tau_i(G)$ is given by 
{\small
\begin{align}
    \mathbb{P}[\tau_i(G)=r] =
    \begin{cases}
        \displaystyle \left(\frac{n-1}{n}\right)^{r-1} \frac{1}{n} & \mbox{if } r < k, \\[10pt]
        \displaystyle \left( \sum_{l=1}^{k-2} \binom{n-1}{l} {r-1 \brace l} l! 
        + \binom{n-1}{k-1} {r-1 \brace k-1} (n-(k-1))(k-1)! \right) 
        \frac{1}{n^{r}} & \mbox{if } r \geq k.
    \end{cases}
\end{align}}
\end{proposition}
\begin{proof}
Assume first $r<k$.
Since any $r$ columns of $G$ are linearly independent, the only way to obtain the $i$-th information strand in the span of the selected columns, is to select any of the $n-1$ columns different from the $i$-th column in the first $r-1$ selections, and the $i$-th column in the $r$-th selection.
Assume now $r\geq k$. 
If in the first $r-1$ draws we selected $l\leq k-2$ distinct columns of $G$ (necessarily all distinct from the $i$-th column $G^i$), then the only way for $G^i$ to be in the span of length $r$ list is to draw $G^i$ as the $r$-th selection.
If in the first $r-1$ draws we selected $l\leq k-2$ distinct columns of $G$ (necessarily all distinct from the $i$-th column $G^i$), then for any choice of a column distinct from the $k-1$ already selected as $r$-th draw, $G^i$ will be in the span of length $r$ list.
The number of ordered lists of length $r$ with $l\leq k-1$ distinct elements, taken from a set of $n-1$ elements, is given by the number of ways of choosing the $l$ elements times the number of ordered assignments of the $r-1$ positions to the $l$ elements, with each element taken at least once (that is, ${r-1 \brace l}l!$).
This completes the proof.
\end{proof}
\begin{remark}
When we compare Proposition~\ref{MDS_density} with~\cite[Lemma 2]{BarLev2025CoverYB} (the case when no coding is used), we see that for $r<k$, the encoding with an MDS code has no advantage with respect to the encoding with the identity code. For $r\geq k$, the values and the probability mass functions of the MDS codes with parameters $[7,3,5]_8$, $[7,4,4]_8$, and the identity code are shown in Table~\ref{tab:PDF} and Figure~\ref{fig:PDF}, respectively.
\end{remark}

\subsection{Hamming codes} 
To analyze Hamming codes, we start by observing that the $\alpha_i(G,s)$'s can be obtained from other parameters introduced in \cite[Lemma 2]{GruicaBarLevRavagnaniYaakobi2024}.
For $i\in \{1, \ldots, k\}$, let $R(i) = \{R_1, ..., R_L\}$ be the collection of minimal recovery sets for the $i$-th encoded strand, where $L$ is their number. We denote by $\xi_i(j, s)$ the number of subsets $S \subseteq \{1, \ldots, L\}$ of cardinality $j$ such that $\bigcup_{h \in S} R_h$ has cardinality $s$.

In other words,
\[\xi_i(j, s) = \left| \left\{S \subseteq [L] : |S| = j, \left| \bigcup_{h \in S} R_h \right| = s\right\} \right|.\]
Note that we denote by $\xi_i(j, s)$ what in \cite{GruicaBarLevRavagnaniYaakobi2024} is denoted by $\tilde{\beta_i}(j, s)$, to avoid a notation clash within our paper.

\begin{proposition}\label{alpha from xi}
Let $G$ be an $k \times n$ matrix in systematic form. We have
\[
\myalpha_i(G,s)=\sum_{j=1}^L\sum_{t=1}^s \binom{n-t}{s-t} (-1)^{j+1}\xi_i(j,t).
\]
\end{proposition}
\begin{proof}
Recall that the quantity $\alpha_i(G,j)$ counts the number of recovery sets of size $j$. 
The claim is obtained using inclusion-exclusion on the number of minimal recovery sets contained in a recovery set of size $s$. 
\end{proof}

 The $\xi_i(j,s)$'s for Hamming codes have been computed in \cite{GruicaBarLevRavagnaniYaakobi2024}.
\begin{proposition}[\text{\cite[Proposition 2]{GruicaBarLevRavagnaniYaakobi2024}}]\label{ksi}
Let $G$ be the  generator matrix in systematic form of a $q$-ary Hamming code of length $n$ and redundancy $k$.
For $i \in \{1,\dots,n\}$, $j \in \{1,\dots,q^{k-1} + 1\}$, and $s \in \{1,\dots,(q^k - 1)/(q - 1)\}$ we have
\[
\xi_i(j, s) = \begin{cases}
\gamma(j, v) & \text{if } s = (q^k - q^v)/(q - 1) - 1, \\
\gamma(j - 1, v) & \text{if } s = (q^k - q^v)/(q - 1),
\end{cases}
\]
where
\[
\gamma(j, v) = \begin{bmatrix}
k - 1 \\
v
\end{bmatrix}_q \, 
\sum_{u=v}^{k-1} q^u
\begin{bmatrix}
k - v - 1 \\
u - v
\end{bmatrix}_q
\binom{q^{k-u-1}}{j}
(-1)^{u-v} q^{\binom{u-v}{2}},
\]
and the symbols in square brackets are $q$-binomial coefficients; see~\cite[Definition 3.1]{Andrews84}.
\end{proposition}
For a $q$-ary Hamming code, the number of minimal recovery sets for each given strand is $L=q^{k-1}+1$. 
The following result is therefore a corollary of Theorem~\ref{thm: probability density function} and Proposition \ref{alpha from xi}.

\begin{corollary}\label{pdfhamming}
Let $G$ be the  generator matrix in systematic form of a $q$-ary Hamming code of length $n$ and redundancy $k$. The probability mass function of the random variable $\tau_i(G)$ is given by 
\begin{align*}
\mathbb{P}[\tau_i(G)=r]=\left(\sum \limits_{s=1}^{r}  \left[ {r \brace s}-n {r-1 \brace s } \right]s!\sum_{j=1}^{q^{k-1}+1}\sum_{t=1}^s \binom{n-t}{s-t} (-1)^{j+1}\xi_i(j,t)\right)\frac{1}{n^r}.
\end{align*}
\end{corollary}
\subsection{Simplex code}
We conclude by computing the mass function associated with simplex codes.
In~\cite{GruicaBarLevRavagnaniYaakobi2024}, the authors compute their expected value $\mathbb{E}[\tau_i(G)]$ by computing the $\alpha_i(G,s)$'s as follows.

\begin{proposition}\cite[Proposition 3]{GruicaBarLevRavagnaniYaakobi2024}
Let $G$ be an $k \times n$ matrix written in a systematic form that generates a $q$-ary simplex code. For all $i, s \in \{1,\dots,n\}$, we have  
\[
\alpha_i(G,s)= \sum_{d=1}^{s} \gauss{ k-1}{d-1}{q}  
\sum \limits_{r=1}^{d} \gauss{ d}{r}{q}   
\binom{\frac{q^r - 1}{q-1}}{s} (-1)^{d-r} q^{\binom{d-r}{2}}.
\]
\end{proposition}

 Using Theorem~\ref{thm: probability density function}, we obtain an explicit formula for the probability mass function of the $q$-ary Simplex code.

\begin{corollary}\label{simplex codes probability density function}
Let $G$ be an $k \times n$ matrix in systematic form that generates the $q$-ary simplex code. The probability mass function of the random variable $\tau_i(G)$ is given by
\begin{align*}
\mathbb{P}[\tau_i(G)=r]=\left(\displaystyle \sum \limits_{s=1}^{r} \left[ {r \brace s}-n {r-1 \brace s } \right] s! \sum \limits_{d=1}^{s}  \gauss{k-1}{d-1}{q}  \sum \limits_{j=1}^{d}  \gauss{d}{j}{q}  \binom{\frac{q^j - 1}{q-1}}{s}  (-1)^{d-j}  q^{\binom{d-j}{2}}\right)\frac{1}{n^r}.
\end{align*}
\end{corollary}

\subsection{Comparisons}
To illustrate Theorems \ref{p-th moment general formula} and \ref{thm: probability density function}, we provide a comparison of the following codes: the $[7,7]_2$ identity code generated by a $7\times 7$ identity; a systematic $[7,3,5]_8$ MDS code; a systematic $[7,4,4]_8$ MDS code; a systematic $[7,4,3]_2$ Hamming code and the $[7,3,4]_2$ Simplex code.

\begin{figure}[h!] \label{PDFplots}
\centering 
\includegraphics[width=0.7\textwidth]{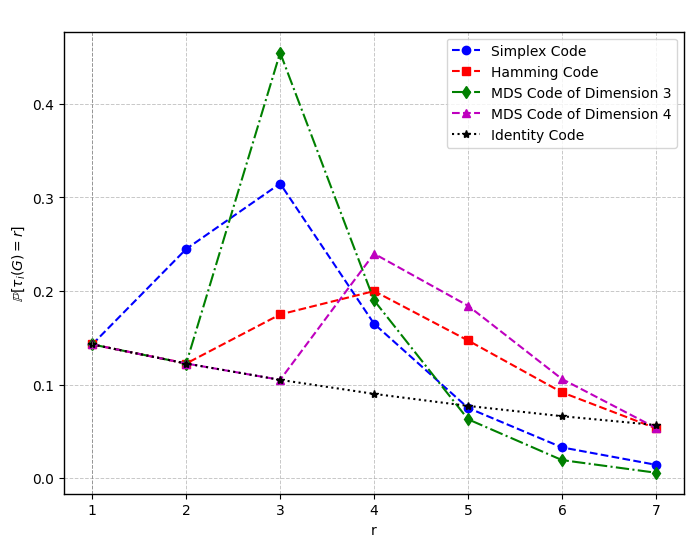} 
\caption{Probability mass functions of Simplex, Hamming, MDS, and identity codes.}
\label{fig:PDF}
\end{figure}

\begin{table}[h!] \label{PDFvalues}
\centering
\begin{tabular}{|c|c|c||c|c||c|}
\hline
\multicolumn{3}{|c||}{$k=3$  } & \multicolumn{2}{|c||}{$k=4$}& \multicolumn{1}{|c|}{$k=7$}\\ \hline
$r$ & MDS Code &Simplex Code & MDS Code& Hamming Code & Identity Code  \\
\hline
1	& 0.143 & 0.143	&	0.143&	0.143 & 0.143	\\ \hline
2	& 0.122 & 0.245	&	0.122&  0.122 &	0.122 \\ \hline
3	& 0.455 & 0.315	&	0.105&  0.175 & 0.105 \\ \hline
4	& 0.190 & 0.165	&	0.240&	0.200 &	0.090\\ \hline
5	& 0.063 & 0.075	&	0.184&  0.147 &	0.077\\ \hline
6	& 0.019 & 0.033	&	0.106&  0.092 &	0.066 \\ \hline
7	& 0.006 & 0.014	&	0.054&  0.053 &	 0.057\\ \hline
\end{tabular}
 \caption{Probability mass functions  $\mathbb{P}[\tau_i(G)=r]$.\label{tab:PDF}}
\end{table}

Fig.~\ref{fig:PDF} shows the probability mass functions of these five codes.
Table~\ref{tab:PDF} contains the values of the probability mass function $\mathbb{P}[\tau_i(G)=r]$ for $r=1,\ldots,7$ for the five codes.

For $r \in\{1, 2\}$, the MDS code of dimension $3$ gives the same values as the identity code. Also, for $r \in \{1, 2, 3\}$, the MDS code of dimension $4$ matches the identity code. Similarly, when $r \in \{1, 2\}$, the Hamming code also gives the same results as the identity code. Finally, Table~\ref{tab:moments} has the variance and the first four moments for each of the five codes.

\prp{\begin{table}[h!]
    \centering
    \begin{tabular}{|l|c|c||c|c||c|}
        \hline
        \multicolumn{3}{|c||}{$k=3$  } & \multicolumn{2}{|c||}{$k=4$}& \multicolumn{1}{|c|}{$k=7$}\\ \hline
        &MDS Code &  Simplex Code & MDS Code &  Hamming Code& Identity Code\\
        \hline
Variance	&	1.467	&	2.167	&	4.100	&	5.033	& 42.000\\ \hline
1st Moment	&	3.000	&	3.000	&	4.000	&	4.000	& 7.000\\ \hline
2nd Moment	&	10.467	&	11.167	&	20.100	&	21.033	&91.000\\ \hline
3rd Moment	&	31.293	&	39.458	&	96.245	&	113.665	&1663.000\\ \hline
4th Moment	&	90.423	&	151.014	&	472.261		&	691.369	&40390.429\\ \hline
    \end{tabular}
    \caption{Comparison of moments for different codes.\label{tab:moments}}    
\end{table}
}

\bibliographystyle{abbrvurl}

\bibliography{Bibliography}

\end{document}